\documentclass{article}%
\usepackage{amssymb,dsfont}
\usepackage{amsmath,color}%
\usepackage{amsfonts}%
\usepackage{graphicx,subfigure}
\usepackage{authblk} 
\providecommand{\U}[1]{\protect\rule{.1in}{.1in}}
\newtheorem{theorem}{Theorem}

\newtheorem{definition}[theorem]{Definition}

\newtheorem{lemma}[theorem]{Lemma}
\newtheorem{notation}[theorem]{Notation}

\newenvironment{proof}[1][Proof]{\noindent\textbf{#1} }{\ \rule{0.5em}{0.5em}}

\newcommand{\beq} {\begin{eqnarray*}}
\newcommand{\eeq} {\end{eqnarray*}}

\title{Semi-Parametric estimation of plane similarities
Application to fast computation of aeronautic loads} 
\date{} 
\author[1,2,3]{Edouard Fournier\thanks{ \texttt{edouard.fournier@airbus.com } }}\author[2]{  St\'ephane Grihon\thanks{ \texttt{stephane.grihon@airbus.com } }} \author[1,3]{Thierry Klein\thanks{ \texttt{thierry.klein@math.univ-toulouse.fr  or thierry01.klein@enac.fr} } } 
\affil[1]{Institut de math\'ematique, UMR5219; Universit\'e de Toulouse;  CNRS, UPS IMT, F-31062 Toulouse Cedex 9, France} 
\affil[2]{Airbus France 316, Route de Bayonne, Toulouse France} 
  \affil[3]{ENAC - Ecole Nationale de l'Aviation Civile, Universit\'e de Toulouse, France}

\begin{document}

\maketitle 

\begin{abstract}
In the big data era, one has often  to conduct the problem of parcimonious data representation. 
In this paper, the data under study are curves and the sparse representation stands on a semi-parametric model. 
Indeed, we propose an original registration model for noisy curves. The model is built transforming an unknown function by plane similarities.  We develop a statistical method that allows to estimate the parameters characterizing  the  plane similarities. The properties of the statistical procedure is studied. We show the convergence and the asymptotic normality of the estimators. Numerical simulations and a real life aeronautic example  illustrate and demonstrate the strength of our methodology.\\

\noindent
\begin{it}Keywords:\end{it} Semiparametric model, Registration of curves, Statistical learning of a physical system \\

\noindent
\begin{it} MSC Classification:\end{it} 62F12, 62F30,  62P30
\end{abstract}

\section{Introduction}
\label{intro}

It may be useful when dealing with a large set of curves differing slightly from one to other to  provide a smart representation. Indeed, usually a reduction taking into account some prior knowledge on the curves may lead to a better understanding of the variability within the population. One way to perform such reduction consists in considering that the set of curves has been obtained by deforming the template. This point of view is usually called curve registration and has been widely studied in statistics  (see  for example   \cite{gasser1995,golubev1988,hardle1990,kneip1995,kneip2000,lawton1972} and references therein).  More recently, this topic has found a second wind pushed by applications in image and signal processing.  We refer for example to \cite{grenander93,McGuire1998ParameterRecovery,
allassonniere2013,gamboa2007,bigot2010,bigot09,gassiat2006,vimond2006} 
for models and techniques related to signal processing problems.
The set of transformations often consists in a parametric family of operator acting on curves \cite{bigot2012}.  One very popular of such model is the so-called shape invariant model (SIM) introduced in \cite{lawton1972} and studied for example in \cite{kneip1995, lindstrom1995,hardle1990,ke2001,gamboa2007}. 
Multidimensional extension of SIM has been studied for functions defined on the plane. In this case, the parametric transformations involve  
rotation and scaling parameters. This model has been studied for image registration see for example \cite{bigot09} or  \cite{McGuire1998ParameterRecovery}. For non parametric approaches, we refer to \cite{ramsay} and references therein.
Generally, the parametric transformation on the curve acts independently on the argument and on the value of the template function. More precisely, let $\tilde{f}$ denotes the template function and $\theta$ denotes the variable parameterizing the transformation. Then,
the transformed function evaluated  at $x$  may be written as 
$T_{2,\theta}(f(T_{1,\theta}(x))$ where $T_{2,\theta}$ is an application on 
$\mathbb{R}$ and $T_{1,\theta}$ acts on the set where $\tilde{f}$ is defined. 
The parametric estimation  is then performed using some $M$-estimation method (see \cite{vdV00}). In this paper, we work on functions defined on $[0,1]$ and  consider a somehow different class of parametric transformations. Indeed, we consider transformations acting jointly on the argument and on the value of the template function. We work with the plane similarities acting on the whole curve 
$\tilde{C}:= (x,\tilde{f}(x))_{x\in[0,1]}$.\\

\noindent
In our paper, we consider first the case where the template is known on the whole design space.  Then, we extend  the results to the more realistic set up where curves and the template are only observed on a grid. The estimations of the unknown parameters are performed using a $M$-estimation procedure.\\

\noindent
Our paper is organized as follow: in Section \ref{sec:1}, we define our model, we develop  $M$-estimation techniques to perform estimation and study the asymptotic behaviour of the estimators. Section \ref{sec:2} is devoted to examples.  We first give a toy model example validating our procedure.  Then, we apply the methodology to an aeronautic model: the prediction of aeronautic loads (see \cite{fournier}). \\

\noindent
All the proofs are postponed to the last section.

\section{Framework, model and analytic results}
\label{sec:1}
In this section, we describe the statistical model studied and give the asymptotic behavior of the $M$-estimators of the unknown parameters.

\subsection{The observations}
\label{sec:11}

\begin{notation} Let be $x \in [0,1]$, and $f: [0,1]\to \mathbb{R}^{+}$. We denote by $C$ the curve\\
    $C:=\begin{pmatrix} x \\ f(x) \end{pmatrix}_{x\in [0,1]}$.
\end{notation}

In our framework, we have at hand $K+1$ curves. $\tilde{C}$ is the reference curve and we assume that the $K$ other curves $C_{j},\ j=1,...,K$ are the images of $\tilde{C}$ by the transformation model described bellow. The $K$ curves are observed on the same random grid $\mathcal{D}_{N}:=\{X_{1},...,X_{N}\}$ where $(X_{i})_{i=1,...,N}$ are iid random variables with uniform distribution. Hence we have at our disposal\\

\begin{center}
  \begin{tabular}{l}
    $C_{j}^{N}=(X_{i},f_{j}(X_{i}))_{i=1,...,N,\ j=1,...,K}$.\\
  \end{tabular}
\end{center}

We will consider the following two cases\\
\begin{itemize}
    \item[] i) $\tilde{C}$ is known everywhere: $\tilde{C}=(x,\tilde{f}(x))$, $\forall x \in [0,1]$,\\
    \item[] ii) $\tilde{C}$ is only known on $\mathcal{D}_{N}$: $\tilde{C}=(x,\tilde{f}(x))$, $\forall x \in \mathcal{D}_{N}$.\\
\end{itemize}

\subsection{Transformation model}
\label{sec:12}
Before defining the transformation model linking the $K$ observed curves to the pattern $\tilde{C}$, one must ensure that the functions $f_{j}$, $j=1,...,K$ and $\tilde{f}$ are "admissible". This is the aim of the next definition:\\

\begin{definition}Let $\theta_{0} \in ]0,\frac{\pi}{2}[$, $\mathcal{F}_{\theta_{0}}$ is the set of applications defined by:
\begin{center}
    $\mathcal{F}_{\theta_{0}}:=\{f:[0,1]\to \mathbb{R}^{+}, f\ \mathrm{is}\ \mathrm{differentiable}\ \mathrm{on}\ [0,1],$\\
    $f(0)>0, f(1)=0$,\\ 
    $f'(x)<0,f'(1)=0,  f'(0)>-\cot{\theta_{0}}\}$.
\end{center}
\end{definition}

Let now define the parametric family of transformation that we will put in action.\\

\begin{notation}For any function $T_{\alpha}:  \mathbb{R}^{2} \to \mathbb{R}^{2}$ depending on a parameter $\alpha\in\Theta$, we will denote by $T_{\alpha}^{1}$ and $T_{\alpha}^{2}$ its two coordinates.\\
\end{notation}

In the following, we define our transformation model.\\

Set $0<\theta_{0}<\theta_{1}<\frac{\pi}{2}$, and $0<\lambda_{min}<\lambda_{max}$. For any $\alpha:=(\theta,\lambda) \in \Theta= [-\theta_{0},\theta_{1}]\times [\lambda_{min},\lambda_{max}]$, set\\
\begin{center}
    $T_{\alpha}:=H_{\lambda}\circ S_{\theta}\circ R_{\theta}$
\end{center}
the composition of a rotation $R_{\theta}$, a rescaling $S_{\theta}$ applied on the $x$-axis, and a homothetic transformation $H_{\lambda}$ on the $y$-axis. More precisely\\
\begin{itemize}
    \item[$\bullet$] $R_{\theta}$ is the rotation centered in $\begin{pmatrix} 1 \\ 0 \end{pmatrix}$ and of angle $\theta$. That is\\
    \begin{align*}
    R_{\theta}: [0,1]\times\mathbb{R}^{+} &\to [-A_{min},1]\times\mathbb{R} \\ \begin{pmatrix} x \\ y \end{pmatrix} &\to \begin{pmatrix} (x-1)\cos{\theta}-y\sin{\theta}+1 \\ (x-1)\sin{\theta}+y\cos{\theta} \end{pmatrix},
    \end{align*}
    with $A_{min}=\sqrt{1+y(0)^{2}}-1$.\\
    
    \item[$\bullet$] After rotating a curve, depending of the angle $\theta$, the resulting design space of the curve is generally no more [0,1]. Hence, we define the following transformation to re-position the curve on [0,1]:\\
    \begin{align*}
    S_{\theta}: [-A_{min},1]\times\mathbb{R} &\to [0,1]\times\mathbb{R} \\ \begin{pmatrix} x \\ y \end{pmatrix} &\to \begin{pmatrix} \frac{x-a}{1 -a} \\ y \end{pmatrix},
    \end{align*}
    with $a$ is the real minimum value of the design space after rotating the curve.\\
    \item[$\bullet$] The scaling transformation of parameter $\lambda>0$ acting on the second coordinate of the curve:
    \begin{align*}
    H_{\theta}: [0,1]\times\mathbb{R^{+}} &\to [0,1]\times\mathbb{R}^{+} \\ \begin{pmatrix} x \\ y \end{pmatrix} &\to \begin{pmatrix} x \\ \lambda y \end{pmatrix}.
    \end{align*}
\end{itemize}

We now introduced the assumption that will be used:\\
\begin{itemize}
    \item[ ] \textbf{(A1)}: The functions $(f_{j})_{j=1,...,K}$, and $\tilde{f}$ belong to $\mathcal{F}_{\theta_{0}}$.\\
    
     \item[ ] \textbf{(A2)}: For any $\theta \in [-\theta_{0},\theta_{1}]$, for any $x \in [0,1]$ and any $f \in \mathcal{F}_{\theta_{0}}$, the second coordinate of $R_{\theta}$ is positive.\\ 
    
    \item[ ] \textbf{(A3)}: $\Theta=[-\theta_{0},\theta_{1}]\times [\lambda_{min},\lambda_{max}]$.\\
    
    \item[ ] \textbf{(A4)}: $\tilde{C}$ is known on $[0,1]$.\\
    
    \item[ ] \textbf{(A5)}: $\tilde{C}$ is known on the grid $\mathcal{D}_{N}$.\\
    
\end{itemize}

Thus, the transformation model considered is as follows:\\
\begin{align*}
    T_{\alpha}:  [0,1]\times\mathbb{R}^{+} &\to [0,1]\times\mathbb{R}^{+} \\ \begin{pmatrix} x \\ f(x) \end{pmatrix} &\to \begin{pmatrix} \frac{(x-1)\cos{\theta}-f(x)\sin{\theta}}{\cos{ \theta}+f(0)\sin{\theta}}+1  \\ \lambda((x-1)\sin{ \theta}+f(x)\cos{\theta}) \end{pmatrix}.
\end{align*}

\subsection{Regression model}
\label{sec:13}

Recall that we wish to adjust the reference curve $\tilde{C}$ on the other curves $C_{j}$ ($j=1,...,K$) by transformations defined in the previous subsection. Notice that these transformations act on both axis. For any $\alpha$, we want to compare the value of the transformed curve $(T_{\alpha}\tilde{C})(X_{i})$ with $f_{j}(X_{i})$. Since the abscissa points are affected by the transformation, we denote by $X_{i}(\alpha)$ the point such that $T_{\alpha}^{1}(X_{i}(\alpha))=X_{i}$. For that reason, we introduce the following definition:
\begin{definition}We denote by $x(\alpha,g)$ the solution of the equation:
\begin{equation}
    T_{\alpha}^{1}(u,g(u))=x.
\label{eq:eq1}
\end{equation}
\end{definition}

To ease the notation, we finally set $x(\alpha):=x(\alpha,g)$, so $X_{i}(\alpha)=X_{i}(\alpha,\tilde{f})$. We consider the parametric regression model:
\begin{equation}
f_{j}(X_{i})=T_{\alpha_{j}^{*}}^{2}(X_{i}(\alpha_{j}^{*}),\tilde{f}(X_{i}(\alpha_{j}^{*})))+\epsilon_{j,i}, (j=1,...,K).
\label{eq:eq2}
\end{equation}

Where:
\begin{itemize}
    \item[$\bullet$] $(X_{i})$ are iid with distribution $\mathbb{P}_{X}$. It is the design on which we observe the curves $C_{j}$;
    \item[$\bullet$] $\alpha_{j}^{*}=(\theta_{j}^{*}$,$\lambda_{j}^{*})$ is the couple of true parameters for each curve ($j=1,...,K$);
    \item[$\bullet$] $\epsilon_{j,i},\ \forall i=1,...,N, \ \forall j=1,...,K$ are iid $\mathcal{N}(0,\sigma^{2})$ random variables. These variables are assumed to be independent of $X_{i}$.\\
\end{itemize}

\subsection{Estimation}
\label{sec:14}
\subsubsection{Estimation when $\tilde{C}$ is known on $[0,1]$}
\label{sec:141}

For the sake of simplicity, let us fix $j$. Relying on a classical $M$-estimation procedure, we consider a semi-parametric method to estimate the parameters and define consequently the following empirical contrast function to fit the reference curve $\tilde{C}$ to  $C_{j}$ ($j=1,...,K$):

\begin{equation}
\begin{split}
M^{j}_{N}(\alpha)&=\frac{1}{N}\sum_{i=1}^{N}(f_{j}(X_{i})-T_{\alpha}^{2}(X_{i}(\alpha),\tilde{f}(X_{i}(\alpha))))^{2}\\
&=\frac{1}{N}\sum_{i=1}^{N}m_{\alpha}^{j}(X_{i}).
\end{split}
\label{eq:eq3}
\end{equation}

The random function $M^{j}_{N}$ is non negative. Furthermore, intuitively, its minimum value
should be reached close to the true parameter $\alpha_{j}^{*}$. Indeed, the following theorem gives the consistency of the $M$-estimator, defined by :
\begin{equation}
\hat{\alpha}^{j}_{N}=\underset{\alpha \in \Theta}{\mathrm{arg min}}\ M^{j}_{N}(\alpha).\\
\label{eq:eq4}
\end{equation}

Recall that our empirical contrast function enters in the general theory of $M$-estimator. The Central Limit Theorem will be shown by using $M$-estimator arguments.\\

\begin{theorem}
Assume that A1, A2, A3 and A4 are satisfied. Then 

\begin{align}
\begin{split}\label{eq:eq5}
    i)\ {}& \hat{\alpha}_{N}^{j} \xrightarrow[N \to +\infty]{\mathbb{P}} \alpha_{j}^{*},\\
\end{split}\\
\begin{split}\label{eq:7}
    ii)\ {}& \sqrt{N}(\hat{\alpha}_{N}^{j}   -\alpha_{j}^{*}) \xrightarrow[N\to + \infty]{\mathcal{L}} \mathcal{N}(0,\Gamma_{\alpha_{j}^{*}}).\\
\end{split}
\end{align}
In particular, the covariance matrix has the following form
\begin{equation}
    \Gamma_{\alpha_{j}^{*}} = V_{\alpha_{j}^{*}}^{-1}2\sigma^{2},
\label{eq:eq8}    
\end{equation}
with $V_{\alpha_{j}^{*}}=2\mathbb{E}[\dot T_{\alpha_{j}^{*}}^{2}\dot T_{\alpha_{j}^{*}}^{2\textbf{T}}]$, and $\dot T_{\alpha_{j}^{*}}^{2}$ is the vector of partial derivatives of $T_{\alpha_{j}}^{2}$ w.r.t elements of $\alpha$ taken at $\alpha_{j}^{*}$.
\end{theorem}

\subsubsection{Estimation when $\tilde{C}$ is observed on $\mathcal{D}$}
\label{sec:142}

In this section, we consider the case where the reference curve $\tilde{C}$ is observed on the same grid $\mathcal{D}:=(X_{i})_{i=1,..,N}$ as the other curve $C_{j}$, i.e $\tilde{C}=\begin{pmatrix} x \\ \tilde{f}(x) \end{pmatrix}_{x \in D}$. By applying the transformation $T_{\alpha}$ to $\tilde{C}$, the transformed pattern $T_{\alpha}\tilde{C}$ is no longer observable on $\mathcal{D}$. As a consequence, one must make use of an approximation process over $\tilde{f}$. Let $\tilde{f}_{N}$ be the linear interpolate of $\tilde{f}$, defined by:\\

\begin{equation}
    \tilde{f}_{N}(x)=\sum_{i=1}^{N}\Delta_{i}(x)\mathds{1}_{x\in[X_{(i)},X_{(i+1)})},
\label{eq:eq9}
\end{equation}

where
\begin{equation}
    \Delta_{i}(x)=\frac{\tilde{f}(X_{(i+1)})-\tilde{f}(X_{(i)})}{X_{(i+1)}-X_{(i)}}\ x\ +\ \tilde{f}(X_{(i)})-\frac{\tilde{f}(X_{(i+1)})-\tilde{f}(X_{(i)})}{X_{(i+1)}-X_{(i)}}X_{(i)}.\\
\label{eq:eq10}    
\end{equation}

It is easy to see that $\tilde{f}_{N}$ belongs also to $\mathcal{F}_{\theta_{0}}$. Replacing $\tilde{C}$ by $\hat{C}$ in (\ref{eq:eq3}) we obtain\\

\begin{center}
$\hat{M}^{j}_{N}(\alpha)=\frac{1}{N}\sum_{i=1}^{N}(f_{j}(X_{i})-T_{\alpha}^{2}(X_{i}(\alpha,N),\tilde{f}_{N}(X_{i}(\alpha,N))))^{2}$,\\
\end{center}

where $X_{i}(\alpha,N)$ is the solution to the equation $T_{\alpha}^{1}(u,\tilde{f}_{N}(u))=X_{i}$.\\

Using the linear interpolate defined by (\ref{eq:eq9}), we show the consistency and asymptotic normality of our $M$-estimator defined as follow:

\begin{equation}
\hat{\hat{\alpha}}^{j}_{N}=\underset{\alpha \in \Theta}{\mathrm{arg min}}\ \hat{M}^{j}_{N}(\alpha).\\
\label{eq:eq32}
\end{equation}

\begin{theorem}
Assume that A1, A2, A3 and A5 are satisfied. Let $\tilde{f}_{N}$ be defined by (\ref{eq:eq9}) and assume that $\exists C>0$ s.t $\forall x\in[0,1]$,  $\tilde{f}_{N}'(x)\leq C$, and $\tilde{f}'(x)\leq C$,  then

\begin{align}
\begin{split}\label{eq:eq33}
    i)\ {}& \hat{\hat{\alpha}}_{N}^{j} \xrightarrow[N \to +\infty]{\mathbb{P}} \alpha_{j}^{*},\\
\end{split}\\
\begin{split}\label{eq:34}
    ii)\ {}& \sqrt{N}(\hat{\hat{\alpha}}_{N}^{j}   -\alpha_{j}^{*}) \xrightarrow[N\to + \infty]{\mathcal{L}} \mathcal{N}(0,\Gamma_{\alpha_{j}^{*}})\\
\end{split}
\end{align}
with $\Gamma_{\alpha_{j}^{*}}$ such defined in (\ref{eq:eq8}).
\end{theorem}

\section{Simulations and applications}
\label{sec:2}
In this section we illustrate the method on numerical applications. The first subsection is dedicated to some simulated toy example while the second to a real problem. The optimisation problems (\ref{eq:eq4}) and (\ref{eq:eq32}) will be numerically solved by using the BFGS algorithm \cite{NoceWrig06}.

\subsection{Simulated toy example}
\label{sec:21}

We consider the following model:\\
\begin{equation*}
f_{\lambda,\theta}(x)=\lambda[(x-1)\sin{\theta}+g(x)\cos{\theta}],    
\end{equation*}

with $g(x)=2(\cos{\pi x}+1)$. We observe $f_{\lambda_{j},\theta_{j}}(x_{i})$ for $i=1,...,100$, for $J=25$ values of $(\lambda,\theta)$ with some iid errors $\epsilon_{ij}$.
\begin{itemize}
    \item[$\bullet$] The observations points $x_{i},\ i=1,...,100$ are iid random variables with uniform distribution on $[0,1]$.
    \item[$\bullet$] The parameters are chosen randomly with the following arbitrary distribution $(\lambda,\theta) \hookrightarrow \mathcal{U}([0,15])\times \mathcal{U}([-\frac{\pi}{3},\frac{\pi}{30}])$.
    \item[$\bullet$] The errors are assumed to be $\mathcal{N}(0,0.01)$.
\end{itemize}

Results are given in Figure \ref{label-simu}. The simulated data are shown in Figure \ref{label-simu} (a). Each of these curves is rescaled to $[0,1]$ to avoid numerical issues. The curve in blue is the reference curve. After the estimation of the parameters $\lambda_{j}$ and $\theta_{j}$, all curves can be rescaled back to their original space as shown in Figure \ref{label-simu} (b). \\

Rescaling the curves allows to easily choose the initial point of parameters for the optimization algorithm taking 1 for $\lambda$ and 0 for $\theta$.\\

\begin{figure}[h!]
\centering
\subfigure[]{\includegraphics[width=6cm]{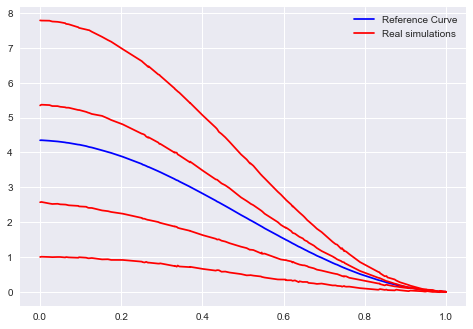}}\subfigure[]{\includegraphics[width=6cm]{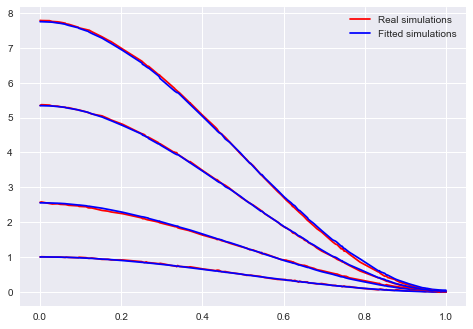}}
\caption{Rotating and scaling results for simulated data: (a) Simulated data; (b) Fitted data}
\label{label-simu}
\end{figure}

\subsection{Aeronautic loads}
\label{sec:22}

An airframe structure is a complex system and its design is a complex task involving today many simulation activities generating massive amounts of data. This is, for example, the process of loads and stress computations of an aircraft. That is the computations of the forces and the mechanical strains suffered by the structure. The overall process exposed in Figure \ref{label-process} is run to identify load cases (i.e aircraft mission and configurations: maneuvers, speed, loading, stiffness...), that are critical in terms of stress endured by the structure and, of course, the parameters which make them critical. The final aim is to size and design the structure (and potentially to reduce loads in order to reduce the weight of the structure). Typically for an overall aircraft structure, millions of load cases can be generated and for each of these load cases millions of structural responses (i.e how structural elements react under such conditions) have to be computed. As a consequence, computational times can be significant.\\

\begin{figure}[h!]
\centering
\includegraphics[width=10cm]{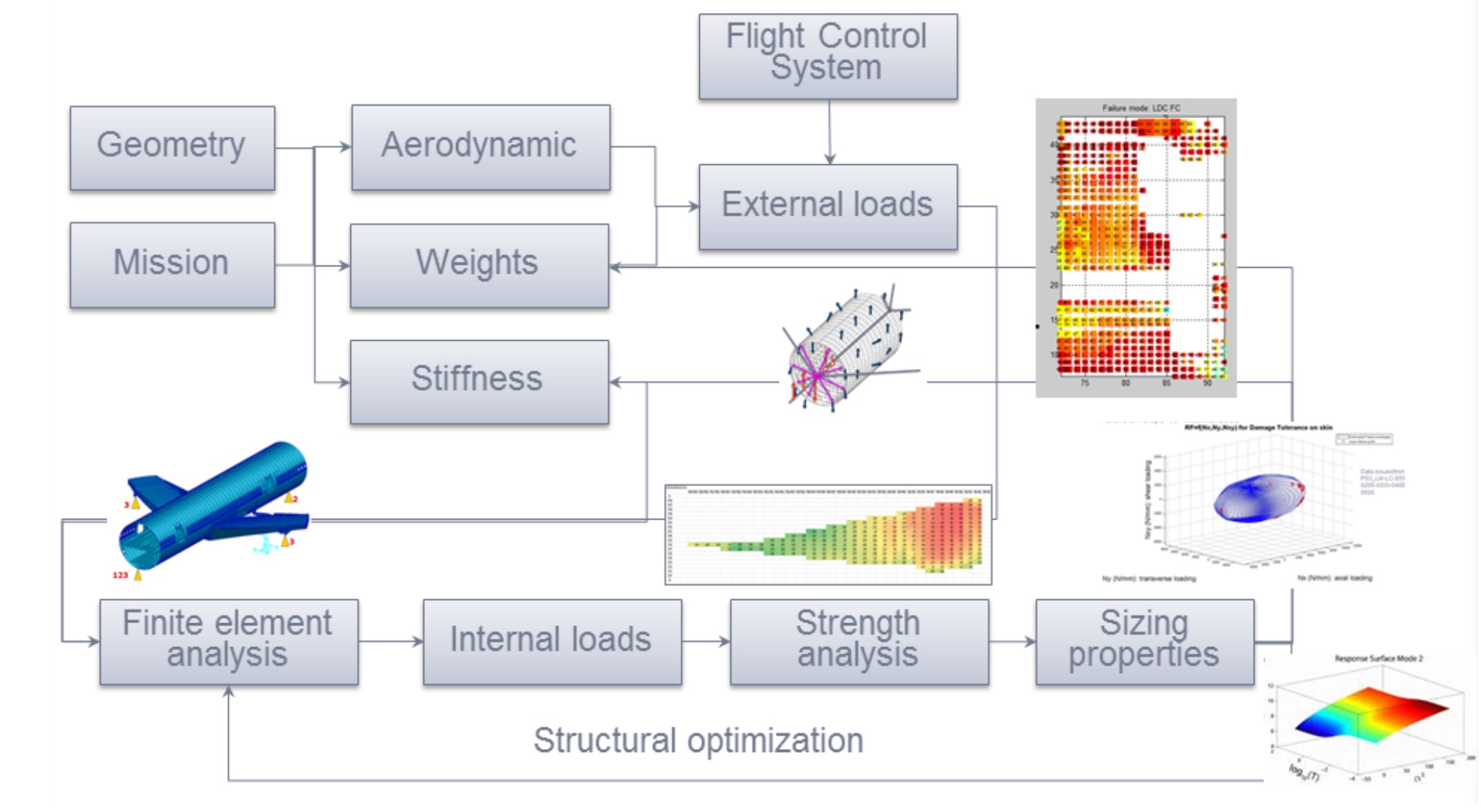}
\caption{Flowchart for loads and stress analysis process}
\label{label-process}
\end{figure}

In an effort to continuously improve methods, tools and ways-of-working, Airbus has invested a lot in digital transformation and the development of infrastructures allowing to treat data (newly or already produced). The main industrial challenge for Airbus is to reduce lead time in the computation and preliminary sizing of an airframe as well as extracting value from already calculated loads. In this paper, we focus on the external loads of a wing: for each load case are calculated the shear forces (transverse forces near to vertical arising from aerodynamic pressure and inertia) and bending moments (resulting from the shear forces, they represent the flexion of the wing) such as shown in Figure \ref{label-BM}.

\begin{figure}[h!]
\centering
\includegraphics[width=5.5cm]{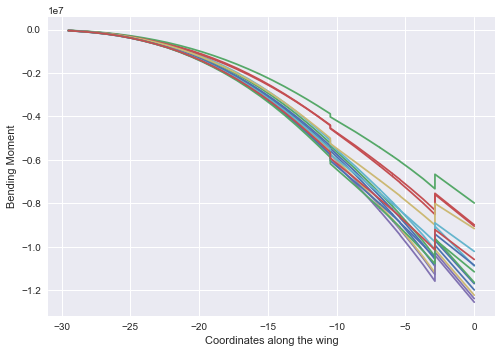}\includegraphics[width=5.5cm]{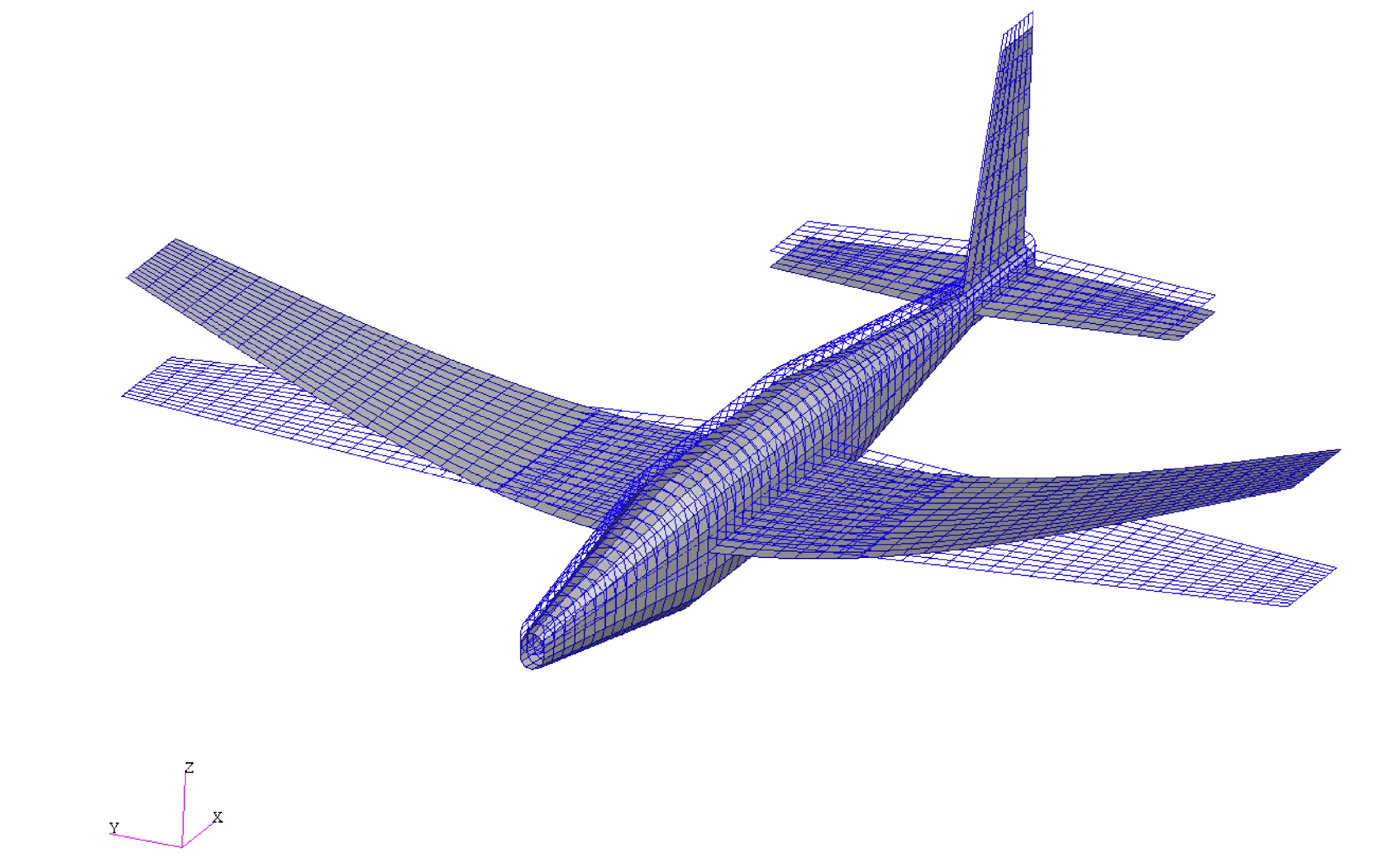}\\
\caption{(a) Examples of bending moments along the wing for different load cases - (b) Finite element model of a generic aircraft representing the wing deformation \cite{ritter01}}
\label{label-BM}
\end{figure}

These external loads appeared to be extremely regular and one can legitimately suppose that it exists a link between all those curves. Indeed, it is natural to assume that it exists a reference bending moment (a reference curve) which can be morphed through a deformation model to give all the other curves.\\

In \cite{fournier}, the authors present an aeronautic model that computes the loads (forces and moments) on the wing of some aircraft model denoted by $ACM1$. They present several statistical methods in order to study these data. In this section, we will compare the method used in \cite{fournier} with the model presented in Section \ref{sec:1} for a new aircraft model called $ACM2$. The data at our disposal represents bending moments of a wing (representing its flexion) of an aircraft calculated for 1152 different configurations (load cases). Each configuration is defined by 28 features (speed of the aircraft, mass, altitude, quantity of fuel, etc.), leading to a bending moment calculated on 45 stations along the wing. In a more formal way, we observe the couple $(X_{j},Y_{j})_{j=1,...,1152}$, where $X_{j}=(X_{j}^{1},...,X_{j}^{28})$ are the features and $Y_{j}=(Y_{j}^{1},...,Y_{j}^{45})$ is the bending moment. The idea is to predict the bending moment for different configurations. The data are represented in Figure \ref{label-210}.\\

\begin{figure}[h!]
\centering
\includegraphics[width=6.5cm]{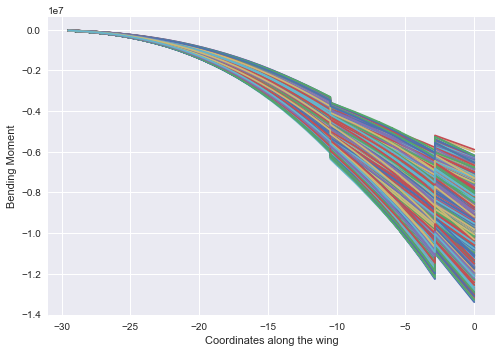}\\
\caption{Representation of all the bending moments of a wing of our data base: the wing root is located at the null coordinate, where the strains are maximum when the wing bends.}
\label{label-210}
\end{figure}

Due to the discontinuities at the $3^{rd}$ and $20^{th}$ stations, we apply our methodology to each section independently. Then, each section can be represented by its minimum and maximum values, and by its rotation and scaling coefficients $\lambda_{j}$ and $\theta_{j}$. Figure \ref{label-matched} assess the quality of the matching process (the reference curve used is the average bending moment).\\

\begin{figure}[h!]
\centering
\includegraphics[width=6.5cm]{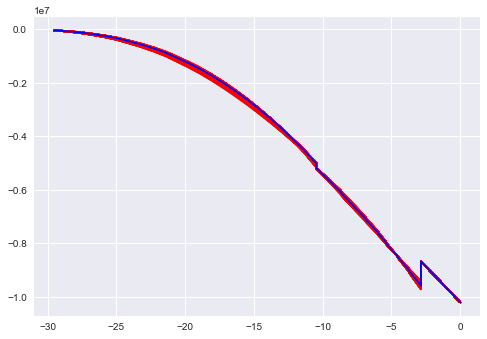}\\
\caption{Results of the matching process}
\label{label-matched}
\end{figure}

Thus the dimensional space of the outputs is reduce to 12 instead of 45. We compare our method to three other methods of \cite{fournier} applied on the outputs: no transformation (we call it raw - we build 45 models, one per station); a PCA (the three first principal components represent 99,9\% of the explained variance - 3 models instead of 45); a polynomial fitting per section (of degree 4 for the first section, of degree 2 for the second and of degree 1 for the third section) which leads to 10 models instead of 45. The Table \ref{table:tabout} sums up the number of outputs to predict depending on the method used.\\

\begin{table}[h!]
\centering
\caption{Number of outputs to be predicted depending on the method used on the raw outputs: Raw, Deformation Model, PCA, polynomial fitting}
        \begin{tabular}{|c|c|c|c|c|c|c|}
          \hline
         & Number of outputs & Names of outputs \\
          \hline
          \hline
        Raw & 45 & Bending moment value at \\
        & &station 0 to 44 \\
           \hline
        Deformation & 12 & $\theta_{1}, \theta_{2}, \theta_{3}, \lambda_{1}, \lambda_{2}, \lambda_{3}, min_{1}, min_{2}, $\\
        Model & &$min_{3}, max_{1}, max_{2}, max_{3}$  \\
           \hline
        PCA & 3 & Principal components 1 to 3 \\
            \hline
        Polynomial fitting & 10 & Coefficients of polynomials \\
            \hline
        \end{tabular}\\
    \label{table:tabout}
\end{table}

The significant advantage of the reduction dimension techniques used is that the response of the model would have a physical form contrary to the simple linear models performed on the raw data. To build our models, we use the Orthogonal Greedy Algorithm (OGA) also known as the Matching Pursuit Algorithm. Detailed explanations can be found in \cite{barron2008}, \cite{sancetta2016} and \cite{mallat93}. Roughly speaking, we consider the problem of approximating a function by a sparse linear combination of inputs.\\

To assess the goodness of fit of our models, we defined for a curve of bending moment $j$ the error rate as follows:\\

\begin{center}
$error(j)=\sqrt{\frac{\sum_{i=1}^{45}(\hat{y}(x_{i})-y_{j}(x_{i}))^{2}}{\sum_{i=1}^{45}y_{j}^{2}(x_{i})}}$, $j=1,...,n_{test}$,\\
\end{center}

where $n_{test}$ is the size of the sample of test. We compute the error rates on (the sample of test is of 25\% the size of the total database). It gives an idea of how far our predictions are. For this standpoint, we can easily compute the empirical cumulative distribution function (CDF): $\forall\ j=1,...,n_{test}$, let $\alpha \in [0,1]$. The empirical CDF is defined as:\\

\begin{center}
    $\alpha \to G(\alpha)=\frac{1}{n}\sum_{j=1}^{n_{test}}\mathds{1}_{(error(j)\leq\alpha)}$
\end{center}

In Table \ref{table:tab6}, we give the values of $G(\alpha)$ for $\alpha=1\%,2\%,5\%,10\%$ and the mean error. In Figure \ref{label-CDF} we give the plots of the function $G(\alpha)$ for the different methods.\\

\begin{table}[h!]
\caption{Average estimated $\mathbb{P}(error\leq 1\%)$, $\mathbb{P}(error\leq 2\%)$, $\mathbb{P}(error\leq 5\%)$ $\mathbb{P}(error\leq 10\%)$, $\mathbb{E}(error)$ calculated on several random test data set (25\% of the size of the total dataset)}
\makebox[\textwidth][c]{
\begin{tabular}{|c|c|c|c|c||}
  \hline
 &Deformation Model &Polynomial Fitting &PCA &Raw\\
 \hline
 \hline
$\mathbb{P}(error\leq 1\%)$  & 17\% & 14\% &16\% &15\%  \\
 \hline
$\mathbb{P}(error\leq 2\%)$  & 45\% & 45\% &43\% &51\%  \\
 \hline
 $\mathbb{P}(error\leq 5\%)$  & 88\% & 88\% &86\% &88\%  \\
 \hline
$\mathbb{P}(error\leq 10\%)$  & 98\% &97\%  &95\% &98\% \\
 \hline
$\mathbb{E}(error)$  & 2.9\% & 2.9\% & 3\% & 2.8\%\\
 \hline
 \hline
\end{tabular}}
\label{table:tab6}
\end{table}

\begin{figure}[h!]
\centering
\includegraphics[width=6.7cm]{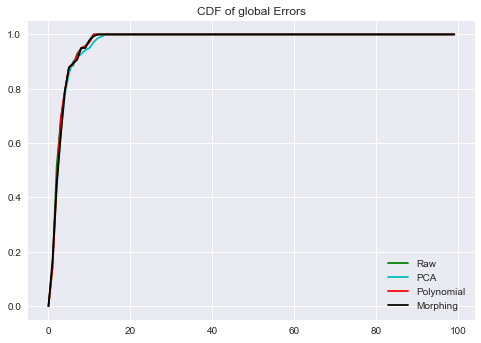}\\
\caption{Empirical CDF of error rates ($\mathbb{P}(error\leq \alpha)$)}
\label{label-CDF}
\end{figure}

Concerning the approximation and prediction of loads, our model is equivalent in average to other tested methods, there are just slightly more observations with an error below 1\%. Nevertheless, in our case, the linear models built through the deformation model are sparser than the other. Indeed, in average, 11 variables are chosen as optimal parameter of the greedy algorithm by cross-validation for the deformation model, 13 for the polynomial fitting, 15 for the PCA and 14 for the raw outputs one.\\

Even though the prediction of loads with the deformation model is so likely equivalent to none transformation, it obtains better results than the polynomial fitting and the PCA. Besides, it is important to notice that it gives to engineers a physical interpretation and idea of how react the wing to new constraints. Besides, using this deformation model gives a physical response contrary to a simple linear model per station whose response could be irregular.\\

\section{Proofs and technical result}
\label{sec:3}
\subsection{Technical result}
\label{sec:31}
This section is dedicated to the technical result used in the proof of Theorem 4.\\

\begin{lemma}Let $X_{1},...,X_{N}$ be $N$ independent and identically distributed random variables with uniform distribution on $[0,1]$ and let $X_{(1)}\leq...\leq X_{(N)}$ be the reordered sample . Let $a_{N}=O(\sqrt{N})$, then:

\begin{equation*}
\mathbb{P}(a_{N}\underset{j}{\sup}|X_{(j+1)}-X_{(j)}|\geq \epsilon)\xrightarrow[N \to +\infty]{} 0.
\end{equation*}

\end{lemma}

\begin{proof}[Proof of Lemma 1] Let $Z_{1},...,Z_{N+1}$ be $N$ independent and identically distributed random variables with exponential distribution with parameter $1$. It is a well known fact that \\
$(\frac{Z_{1}}{\sum_{k=1}^{N+1}Z_{k}},\frac{Z_{1}+Z_{2}}{\sum_{k=1}^{N+1}Z_{k}},...,\frac{Z_{1}+...+Z_{N}}{\sum_{k=1}^{N+1}Z_{k}})\overset{(\mathcal{L})}{=}(X_{(1)},...,X_{(N)})$ and we have 
\begin{equation}
X_{(j+1)}-X_{(j)}\overset{(\mathcal{L})}{=}\frac{Z_{j+1}}{\sum_{k}Z_{k}}.
\label{eq:ordered}
\end{equation}

Now, for $\epsilon >0$

\begin{equation*}
\begin{split}
\mathbb{P}(\underset{j}{\sup}|X_{(j+1)}-X_{(j)}|\geq \epsilon) &\leq \sum_{j} \mathbb{P}(X_{(j+1)}-X_{(j)}\geq \epsilon)\\
&\leq N \max_{j} \mathbb{P}(X_{(j+1)}-X_{(j)}\geq \epsilon).\\
\end{split}
\end{equation*}

By using (\ref{eq:ordered}), we have
\begin{equation*}
\begin{split}
\mathbb{P}(X_{(j+1)}-X_{(j)}\geq \epsilon)=(1-\epsilon)^{N-1}.\\
\end{split}    
\end{equation*}

Then,

\begin{equation*}
\begin{split}
\mathbb{P}(\underset{j}{\sup}|X_{(j+1)}-X_{(j)}|\geq \epsilon) &\leq N(1-\epsilon)^{N-1}.\\
\end{split}    
\end{equation*}

The result follows replacing $\epsilon$ by $\frac{\epsilon}{a_{N}}$ and letting $N \to +\infty$.
\end{proof}

\subsection{Proofs of Theorems 3 and 4}
\label{sec:32}

\begin{proof}[Proof of Theorem 3] To ease the notation, we do not display the dependency in $j$.\\

i) By (\ref{eq:eq3}) it is easy to see that $M_{N}(\alpha)$ is an empirical mean of iid bounded random variables. Thus, by the Strong Law of Large Number (SLLN)
\begin{center}
    $M_{N}(\alpha) \xrightarrow[N \to +\infty]{p.s} M(\alpha)$,
\end{center}
with $M(\alpha)=\mathbb{E}[\epsilon^{2}]+ \mathbb{E}[(T_{\alpha}^{2}(X(\alpha),\tilde{f}(X(\alpha)))-T_{\alpha^{*}}^{2}(X(\alpha^{*}),\tilde{f}(X(\alpha^{*}))))^{2}]$.\\

$M(\alpha)$ is continuous and has an obvious unique minimum $\alpha^{*}$. Since $\Theta$ is compact, this implies that $\underset{\alpha:d(\alpha,\alpha^{*})\geq\epsilon}{\inf}\  M(\alpha) >M(\alpha^{*})$ is satisfied (see Problem 27 p. 84 in \cite{vdV00}).\\

It remains to prove that $\{m_{\alpha} : \alpha \in \Theta\}$ is a Glivenko-Cantelli class. Thanks to the remark following the proof of Theorem 5.9 in \cite{vdV00}, this is an easy consequence of the continuity of $\alpha \to m_{\alpha}$ and the fact that the function is bounded by a continuous and integrable function on $[0,1]$. Indeed, it exists at least a function $f^{*}$ in $\mathcal{F}_{\theta_{0}}$ which bounds every other functions, and two constants $K_{1} >0, K_{2}>0$ such that
\begin{center}
    $m_{\alpha}(x) \leq K_{1}(f^{*}(x)+K_{2})^{2}$,
\end{center}

and
\begin{equation}
\underset{\alpha\in\Theta}{\sup}\ |M_{N}(\alpha)-M(\alpha)|\xrightarrow[N \to +\infty]{\mathbb{P}}\ 0.\\
\label{eq:eq6}
\end{equation}

The result follows from the Theorem 5.7 in \cite{vdV00}.\\

ii) The Central Limit Theorem will be a consequence of Theorem 5.23 in \cite{vdV00}. Recall that
\begin{center}
    $m_{\alpha}(x)=[f(x)-\lambda((x(\alpha)-1)\sin \theta + \cos \theta \tilde{f}(x(\alpha)) )]^{2}$. 
\end{center}

By the Implicit Function Theorem, that is easy to see that $\alpha \to x(\alpha)$ is $C^{1}$ on a compact set. This implies that the norm of the gradient of $m_{\alpha}$ is uniformly bounded in $\alpha$. Hence $\exists \dot\phi(x) \in L^{1}$ such that $||\nabla_{\alpha}m_{\alpha}(x)||\leq \dot\phi(x)$ hence\\

\begin{center}
    $|m_{\alpha_{1}}(x)-m_{\alpha_{2}}(x)| \leq \dot\phi(x)\times ||\alpha_{1}-\alpha_{2}||.$
\end{center}

In order to give an explicit formula for the limit variance, we apply the results of Example 5.27 in \cite{vdV00} where $f_{\theta}$ becomes in our case $T_{\alpha}^{2}$ and hence, we have
\begin{equation*}
    \sqrt{N}(\hat{\alpha}_{N}^{j}   -\alpha_{j}^{*}) \xrightarrow[N\to + \infty]{\mathcal{L}} \mathcal{N}(0,\Gamma_{\alpha_{j}^{*}}),\\
\end{equation*}

with $\Gamma_{\alpha_{j}^{*}} = V_{\alpha_{j}^{*}}^{-1}2\sigma^{2}$ and $V_{\alpha_{j}^{*}}=2\mathbb{E}[\dot T_{\alpha_{j}^{*}}^{2}\dot T_{\alpha_{j}^{*}}^{2\textbf{T}}]$.

\end{proof}

\begin{proof}[Proof of Theorem 4] \ \\

i) To prove the consistency of $\hat{\hat{\alpha}}_{N}$ we have to show that\\
\begin{center}
    $\underset{\alpha\in\Theta}{\sup}|\hat{M}_{N}(\alpha)-M(\alpha)|\xrightarrow[N \to +\infty]{\mathbb{P}}\ 0$.\\
\end{center}

We have,
\begin{center}
    $\underset{\alpha\in\Theta}{\sup}|\hat{M}_{N}(\alpha)-M(\alpha)| \leq \underset{\alpha\in\Theta}{\sup}|\hat{M}_{N}(\alpha)-M_{N}(\alpha)| + \underset{\alpha\in\Theta}{\sup}|M_{N}(\alpha)-M(\alpha)|$.\\
\end{center}

It has been shown in the proof of Theorem 3 that
\begin{center}
    $\underset{\alpha\in\Theta}{\sup}|M_{N}(\alpha)-M(\alpha)|\xrightarrow[N \to +\infty]{\mathbb{P}}\ 0$.\\
\end{center} 

It remains to prove that
\begin{center}
    $\underset{\alpha\in\Theta}{\sup}|\hat{M}_{N}(\alpha)-M_{N}(\alpha)|\xrightarrow[N \to +\infty]{\mathbb{P}}\ 0$.\\
\end{center} 

To ease the notation, we write $T_{\alpha}^{2}(i,N)=T_{\alpha}^{2}(X_{i}(\alpha,N),\tilde{f}_{N}(X_{i}(\alpha,N)))$, and $T_{\alpha}^{2}(i)=T_{\alpha}^{2}(X_{i}(\alpha),\tilde{f}(X_{i}(\alpha)))$. Set

\begin{equation*}
\begin{split}
D_{N}(\alpha)&=M_{N}(\alpha)-\hat{M}_{N}(\alpha)\\
&=\frac{1}{N}\sum_{i=1}^{N}2y(X_{i})[T_{\alpha}^{2}(i,N)-T_{\alpha}^{2}(i)]-\frac{1}{N}\sum_{i=1}^{N}[T_{\alpha}^{2}(i,N)-T_{\alpha}^{2}(i)][T_{\alpha}^{2}(i,N)+T_{\alpha}^{2}(i)].\\
\end{split}    
\end{equation*}

As $f$ and $T_{\alpha}^{2}$ are continuous and bounded on $\Theta\times [0,1]$, this implies that:\\
\begin{equation*}
\begin{split}
|D_{n}(\alpha)|&\leq |\frac{1}{N}\sum_{i=1}^{N}2f(X_{i})[T_{\alpha}^{2}(i,N)-T_{\alpha}^{2}(i)]|+|\frac{1}{N}\sum_{i=1}^{N}[T_{\alpha}^{2}(i,N)-T_{\alpha}^{2}(i)][T_{\alpha}^{2}(i,N)+T_{\alpha}^{2}(i)]|\\
&\leq K(\frac{1}{N}\sum_{i=1}^{N}[T_{\alpha}^{2}(i,N)-T_{\alpha}^{2}(i)]^{2})^{\frac{1}{2}}\\
&\leq K'(\frac{1}{N}\sum_{i=1}^{N}[(X_{i}(\alpha,N)-X_{i}(\alpha))(1+C)+(\tilde{f}_{N}(X_{i}(\alpha))-\tilde{f}(X_{i}(\alpha))]^{2})^{\frac{1}{2}}.\\
\end{split}
\end{equation*}

By construction, there exists $j$ such that $X_{(j)}\leq X_{i}(\alpha) \leq X_{(j+1)}$, and $X_{(j)}\leq X_{i}(\alpha,N) \leq X_{(j+1)}$ which leads to:

\begin{equation*}
    X_{i}(\alpha,N)-X_{i}(\alpha)=\gamma(X_{(j+1)}-X_{(j)}).\\
\end{equation*}

Besides, since there exists $\gamma'>0$ such that\\
$\tilde{f}_{N}(X_{i}(\alpha))=\gamma'\tilde{f}(X_{(j+1)})+(1-\gamma')\tilde{f}(X_{(j)})$ we have

\begin{equation*}
\begin{split}
|\tilde{f}_{N}(X_{i}(\alpha))-\tilde{f}(X_{i}(\alpha))|&\leq \gamma'|\tilde{f}(X_{(j+1)})-\tilde{f}(X_{(j)})|\\
&\leq C\gamma'|X_{(j+1)}-X_{(j)}|, \\
\end{split}
\end{equation*}

and

\begin{align*}
    |D_{n}(\alpha)|&\leq K'(\frac{1}{N}\sum_{j=1}^{N}[X_{(j+1)}-X_{(j)}]^{2})^{\frac{1}{2}}  \\ 
    \underset{\alpha\in\Theta}{\sup}|D_{n}(\alpha)|&\leq K'(\frac{1}{N}\sum_{j=1}^{N}\underset{\alpha\in\Theta}{\sup}[X_{(j+1)}-X_{(j)}]^{2})^{\frac{1}{2}}\\
    &\leq K' \underset{j}{\sup}|X_{(j+1)}-X_{(j)}|.\\
\end{align*}

By Lemma 1 $\mathbb{P}(K'\underset{j}{\sup}|X_{(j+1)}-X_{(j)}|\geq \epsilon)\xrightarrow[N \to +\infty]{} 0$. Hence $D_{N}$ is bounded by an integrable and continuous function which goes to 0 in probability on $\Theta$\\

\begin{center}
    $\underset{\alpha\in\Theta}{\sup}|\hat{M}_{N}(\alpha)-M(\alpha)|\xrightarrow[N \to +\infty]{\mathbb{P}}\ 0$.\\
\end{center}

So we may conclude.\\

ii) First, we use that $\sqrt{N}(\hat{\hat{\alpha}}_{N}-\alpha^{*})=\sqrt{N}(\hat{\hat{\alpha}}_{N}-\hat{\alpha}_{N})+\sqrt{N}(\hat{\alpha}_{N}-\alpha^{*})$. By Theorem 3, $\sqrt{N}(\hat{\alpha}_{N}   -\alpha^{*}) \xrightarrow[N\to + \infty]{\mathcal{L}} \mathcal{N}(0,\Gamma_{\alpha^{*}})$ with  $\Gamma_{\alpha^{*}}$ defined in (\ref{eq:7}). It remains to prove that $\sqrt{N}(\hat{\hat{\alpha}}_{N}   -\hat{\alpha}) \xrightarrow[N\to + \infty]{\mathbb{P}} 0$.\\

Using the same arguments as in the proof i), we have\\
\begin{equation}
    \mathbb{P}(\sqrt{N}\underset{\alpha\in\Theta}{\sup}|\hat{M}_{N}(\alpha)-M_{N}(\alpha)|\geq \epsilon)\leq \mathbb{P}(K\sqrt{N}\underset{\alpha\in\Theta}{\sup}|X_{(j+1)}-X_{(j)}|\geq \epsilon)
\label{eq:lemm}
\end{equation}

The right hand side of (\ref{eq:lemm}) converges to $0$ by Lemma 1. This implies that $\sqrt{N}(\hat{\hat{\alpha}}_{N}   -\hat{\alpha}) \xrightarrow[N\to + \infty]{\mathbb{P}} 0$.\\

\end{proof}

\section{Perspectives and conclusion}

One of the main quality of our approach is that it is easy to implement and execute. The cost function being simple, we use a BFGS algorithm to find the optimal parameters, and because of the regularity of curves we deal with, the initial points for optimization can be easily defined.\\

Furthermore, the search of the coordinate of the reference curve which is sent to the coordinate of the curve to fit can be easily implemented with a simple value search.\\

Besides, the deformation parameters can be exploited through an explainable model such as the linear model used in the real world problem.\\

It seems that the deformation model is robust if the noise is controlled. An interesting extension of this work would be to study what is going on when the reference curve is noisy. A generalization of this work to less regular functions would be worthwhile. Finally, it would be interesting to include in the model a way to handle discontinuities in order to reduce the dimension and have a more global representation of the deformation.\\


\bibliographystyle{spmpsci}      
\bibliography{bibli.bib}

\end{document}